\documentclass[submission,copyright,creativecommons]{eptcs}
\providecommand{\event}{}
\usepackage{hyperref,tikz,tikz-cd}
\usepackage{braket, breakurl,amssymb,amsmath,amsthm,mathtools,enumerate,etoolbox}
\usepackage[sort&compress,numbers,comma]{natbib}
\usepackage[toc,page]{appendix}
\title{Almost Equivalent Paradigms of Contextuality}
\author{Linde Wester
\institute{University of Oxford}
\institute{Department of Computer Science}
\email{linde.wester@cs.ox.ac.uk}
}
\def\titlerunning{Almost Equivalent Paradigms of Contextuality}
\def\authorrunning{L. Wester}

\usepackage{thmtools,thm-restate}
\theoremstyle{plain}
\newtheorem{theorem}{Theorem}
\newtheorem{postulate}{Assumption}
\newtheorem{example}{Example}
\newtheorem*{example*}{Example}
\newtheorem{remark}{Remark}
\newtheorem{definition}{Definition}
\newtheorem{corollary}[theorem]{Corollary}
\newtheorem{lemma}[theorem]{Lemma}
\newtheorem{prop}[theorem]{Proposition}
\newtheorem{pos}[postulate]{Assumption}
\theoremstyle{definition}

\newcommand{\cE}{\mathcal{E}}
\newcommand{\cM}{\mathcal{M}}
\newcommand{\Emp}{\mathcal{E}mp}
\newcommand{\Ot}{\mathcal{O}T}
\newcommand{\Or}{\mathcal{O}R}

\newcommand{\R}{R}

\usetikzlibrary{arrows}

\begin{document}
\maketitle

\begin{abstract}
Various frameworks that generalise the notion of contextuality in theories of physics have been proposed; one is the sheaf-theoretic approach by Abramsky and Brandenburger; an other is the equivalence-based approach by Spekkens. 
We show that these frameworks are equivalent for scenarios with preparations and measurements, whenever factorizability is justified. This connection gives rise to a categorical isomorphism between suitable categories. We combine the advantages of the two approaches to derive a canonical method for detecting contextuality in such settings. 

\end{abstract}

\section{Introduction}

\subsection*{Two Formalisms for Contextuality}
\label{thm:goldbach}
Contextuality of quantum mechanics entails the impossibility of assigning predetermined outcomes to observables in a way that is independent of the context or method of observation. It was first described by Kochen and Specker in ~\cite{KS}. Such an assignment would provide a hidden variable model, which is an explanation for observations of the physical world according to the laws of classical mechanics. Hence, the study of contextuality offers a way to specify the manner in which quantum mechanics deviates from the theory of classical mechanics. 

Recently, various formalisms of different scope and nature have been proposed, which generalise the currently known examples of contextuality:~\cite{abramskybrandenburger},~\cite{afls},\cite{csw}, \cite{dzhafarovkujala}, \cite{roumen}, \cite{spekkens}. The relation between some of the formalisms has been studied in~\cite{DeSilva},~\cite{statonuijlen}. In this paper, we unify the sheaf-theoretic formalism by Abramsky and Brandenburger~\cite{abramskybrandenburger} with the equivalence-based notion of contextuality developed by Spekkens in ~\cite{spekkens}. These approaches share the goal of expressing the notion of non-contextuality in a manner that is independent of the quantum formalism.  They are applicable to any operational or empirical theory, which is a high-level description of an experimental setting.


\subsection*{The Sheaf Approach}
In the sheaf approach to contextuality, one defines contextuality as the non-existence of a joint probability distribution over the outcomes of a set of measurements. It is formulated within the mathematical framework of sheaf theory, as the non-existence of a global section for a presheaf of distributions over measurement outcomes. The mathematical framework provides algorithmic methods based on sheaf-cohomology to detect contextuality~\cite{abramskybarbosamansfield}, as well as means of quantifying contextuality~\cite{abramskybarbosamansfield2} and a structural method for deriving non-contextuality inequalities~\cite{abramskyhardy}. The sheaf approach is applicable to measurement-based quantum computing. In~\cite{Raussendorf}, it was shown that any mod-2 nonlinear computation in measurement based quantum computing with a linear classical processor requires sheaf-theoretic contextuality. 

\subsection*{The Equivalence-Based approach}\label{sec:prob}
Contextuality in the equivalence-based approach is defined as the non-existence of certain hidden variable models, called ontological representations, of operational theories. Such an ontological representation must be determined by the statistical data of the experiment only. It cannot depend on any additional data, regarded as the 'context'. This gives a natural explanation for operational equivalence of measurements and preparations: we cannot distinguish them because they correspond to the same ontological values. 
The formalism distinguishes three different types of contextuality: contextuality for preparations, transformations and for (unsharp) measurements. Unlike traditional notions of contextuality, the formalism does not presuppose outcome assignments in a hidden-variable model to be deterministic. In doing so, it provides a framework for contextuality tests which are robust to noise~\citep{kunjwalspekkens}~\cite{mazurek2016experimental}.

\subsection*{Overview}
In this paper, we draw a formal connection between the sheaf-theoretic approach to contextuality on the one hand and general contextuality in the equivalence-based approach on the other hand. General contextuality is defined as the existence of at least one of the three different types of equivalence-based contextuality. 

We expand the scope of the sheaf formalism to ensure that any scenario that can be described in the sheaf formalism corresponds to an operational theory in the equivalence-based formalism and vice versa. We call the type of theories in the extended sheaf-theoretic formalism 'empirical theories'. As the notions of empirical and operational theories are interchangeable, we will often simply call them 'theories'. The joint distribution that is characteristic to a non-contextual empirical theory in the sheaf-theoretic sense gives rise to an ontological model for the corresponding operational theory. A priori, this operational theory may not be non-contextual in the equivalence-based sense, as the joint distribution may depend on data other than the outcome statistics. Therefore, we show that whenever an empirical theory is non-contextual in the sheaf sense, we can eliminate any statistically redundant data to obtain a new 'minimal' non-contextual empirical theory. We use this notion of a minimal theory to construct a non-contextual ontological representation (in the equivalence-based sense) for each non-contextual theory in the sheaf-theoretic sense. We call this the 'canonical' ontological representation of a theory.
Finally, we derive that an operational theory can be realised by a factorizable non-contextual ontological representation in the equivalence-based sense, if and only if it is non-contextual in the sheaf-theoretic sense. This is a generalisation of a result by Abramsky and Brandenburger which was formulated in terms of empirical models and the traditional notion of non-contextuality for ontological models~\cite{abramskybrandenburger}.
As a result, the contextuality argument for preparations and for unsharp measurements given by Spekkens in~\cite{spekkens} can be formulated in the sheaf-theoretic formalism. More generally, contextuality for preparations and unsharp measurements is independent of the chosen formalism for contextuality. 
Furthermore, we generalise the result in~\cite{spekkens} that equivalence-based non-contextual representations of PVM's in quantum theory are outcome-deterministic, to an analogous result for an appropriate notion of sharp measurements in general operational theories. We show that operational theories with such measurements are non-contextual if and only if their 'canonical' representations are non-contextual.

\subsection*{Outline}
In Section~\ref{sec:spek} we recall the equivalence-based approach to contextuality. We discuss how known examples of contextuality arise in this formalism by imposing additional assumptions which force the outcome assignments in an ontological model to be deterministic. We analyse the notion of contextuality beyond theories where outcome-determinism of ontological representations is justified.
As an illustration, we discuss Mermin's All Versus Nothing argument and Bell's scenario in the equivalence-based model. 
In Section~\ref{sec:emp}, we recall the sheaf-theoretic approach. We extend the framework to incorporate preparations. We explore the role of convexity in this formalism and  derive a sheaf-theoretic contextuality proof for the scenario of preparations and unsharp measurements~\cite{spekkens}.
In Section~\ref{sec:uni}, we introduce a method for constructing a canonical non-contextual ontological representation for any non-contextual empirical theory. We prove that such a canonical non-contextual ontological representation exists whenever the theory admits a factorizable non-contextual ontological representation.
 In Section~\ref{sec:catiso} we introduce the categories $\Emp$ of empirical theories, $\Ot$ of operational theories, and $\Or$ of ontological representations. We show how the correspondence between empirical and operational theories gives rise to a categorical isomorphism, which maps the subcategory of non-contextual empirical theories to the subcategory of operational theories that can be realised by a factorizable non-contextual ontological representation. 
In Section~\ref{sec:postlude}, we discuss contextuality for unsharp measurements and give an example for which the two notions of contexuality are different.

\section{The Equivalence-Based Approach}\label{sec:spek}
The notion of contextuality used in the classic examples~\cite{KS},~\cite{hardy},~\cite{mermin}~\cite{bell}, is specific to the Hilbert space formalism for quantum mechanics. The equivalence-based approach provides us with a more general, operational principle for defining contextuality. In this section, we discuss the structure and assumptions needed to derive various examples of contextuality from this principle. This has been done previously in the specific cases of the Kochen-Specker scenario in~\cite{kunjwalspekkens} and Spekkens' contextuality of 2-dimensional quantum systems in~\cite{spekkens}. Here, we define conditions under which the equivalence-based approach gives rise to any contextuality scenario that relies on the impossibility of a hidden variable model with deterministic outcome assignments. As an illustration, we discuss Bell's non-locality scenario and Mermin's all versus nothing argument from the equivalence-based point of view. Furthermore, we analyse the notion of equivalence-based contextuality for the class of theories where outcome-determinism cannot be justified.

\subsection*{Operational Theories}
Firstly, we recall the equivalence-based approach to contextuality ~\cite{spekkens},~\cite{liangspekkenswiseman}. Consider two sets, $P$ and $M$, of preparation procedures and measurement procedures, respectively. For each measurement $m \in M$ there is a set finite set $O^m$ of possible measurement outcomes. For each pair $(p,m) \in P \times M$, there exists a probability distribution $d_{p,m}: O^m \rightarrow [0,1]$ over the set of possible outcomes $O^{m}$. The value $d_{p,m}(k)$ should be understood as the probability of obtaining the outcome $k$ when a preparation $p$ is performed, followed by a measurement $m$. 
We write $D$ for the indexed set of probability distributions $\{d_{p,m}\}_{p \in P, m\in M}$.
An {\bf operational theory} is defined by a tuple $(P, M, D, O)$, where $O=\cup_m O^m$. In \cite{spekkens}, operational theories also contain a set of transformation procedures, but we will not consider these here. We can still account for any transformation followed by a measurement by considering it as a new measurement. However, we lose the significance of compositionality, which can result in additional statistical equivalences.

Preparations and measurements are {\bf statistically equivalent} when they are not distinguishable based on the measurement statistics in the operational theory. Let $p,p' \in P$ be preparations and let $m,m' \in M$ be measurements, this is expressed below. 

\begin{align*}
p &\sim p' &\Leftrightarrow \hspace{1cm}&d_{p,m} = d_{p',m}& \forall m \in M \\
m &\sim m' &\Leftrightarrow \hspace{1cm}&d_{p,m} = d_{p,m'} & \forall p \in P\\ 
(m,k) &\sim (m',k') &\Leftrightarrow \hspace{1cm}&d_{p,m}(k) = d_{p,m'}(k') & \forall p \in P
\end{align*}

\begin{example*}[Quantum Mechanics]
The standard example of an operational theory is the Hilbert space formalism for quantum mechanics.  Equivalence classes of preparation procedures correspond to  density matrices. Equivalence classes of measurement procedures correspond to POVM's. Different decompositions of density matrices or POVM's correspond to different preparation or measurement procedures, respectively. The outcome set of each measurement procedure corresponds to its POVM elements. The indexed set of distribution functions $D$ is derived from the Born rule.
\end{example*}

An {\bf ontological representation} of the operational theory $A=(P,M,D,O)$ consists of a discrete set of ontological values $\Omega$, together with sets of distribution functions $\mu=\{\mu_p: \Omega \rightarrow [0,1]\}_{p \in P}$ and $\xi=\{\xi_{m}(\lambda): O^m \rightarrow[0,1]\}_{\lambda \in \Omega, m \in M}$. The distribution functions are such that they realise the measurement statistics of $A$, which is expressed by the formula below.
 
\begin{equation}\label{def:or}
\sum_{\lambda \in \Omega} \xi_m(\lambda)(k) \mu_p(\lambda)= d_{p,m}(k) \hspace{2cm} \forall p \in P, m \in M
\end{equation}

An ontological representation, like a hidden variable model, should be thought of as representing a physical system as it really is, while the operational theory merely describes our knowledge of the system, which may not be accurate or complete. 

In this paper, we will assume that all ontological values can be obtained by performing some preparation procedure. This means that for all $\lambda \in \Omega$ there exists a preparation $p \in P$, such that $\mu_p(\lambda)>0$.

\begin{definition}
An {\bf ontological representation} is called {\bf preparation non-contextual} if $\mu_{p} = \mu_{p'}$ whenever $p \sim p'$; it is called {\bf measurement non-contextual} if $\xi_{k,m} = \xi_{k',m'}$ whenever $(m,k) \sim (m',k')$; it is called {\bf non-contextual} if it is preparation non-contextual as well as measurement non-contextual. An {\bf operational theory} is called {\bf non-contextual} whenever there exists a non-contextual ontological representation that realises the theory. When there exists no such ontological representation, and operational theory is called contextual.
\end{definition}

Contextuality is characterised as the impossibility of {\bf any} non-contextual hidden variable model. One would therefore want to consider non-discrete infinite spaces of ontological values as well. We conjecture that all results can be generalised to hold for topological measure spaces. We leave this for future work. 

Nevertheless, the notion of non-contextuality without any further restictions is too permissive. We can find a non-contextual ontological representation for any operational theory with discrete sets $P,M$. This covers all classic examples of contextuality mentioned so far. The ontological  representation is defined as follows:
\begin{align*}
\Omega:= \{[p]\}_{p \in P} \hspace{.5cm}
\mu_q([p]):= \delta_{p,q} \hspace{.5cm}  
\xi_m([p])(k):= d_{p,m}(k) 
\end{align*}

Here, $[p]$ is the statistical equivalence class of $p \in P$. 
One could ask wether it is possible to find a non-contextual ontological representation that satisfies specific properties. These can either be derived from the formalism or justified by physical principles. 
In the rest of this section, we discuss which additional properties we need to derive known examples of contextuality from the equivalence-based definition of contextuality.

\subsection*{Deriving Contextuality Arguments from Convex Structure}
The following assumption is required for the contextuality proof of preparations and unsharp measurements of 2-dimensional quantum systems given in~\cite{spekkens}. 

\begin{pos}[Preservation of Convexity]\label{pos:conv}
Let $p, p_1, p_2 \in P$ be preparation procedures, such that $p= c p_1 +(1-c) p_2$ is a convex combination of $p_1$ and $p_2$; let $m, m_1, m_2 \in M$ be measurement procedures, such that $m= c m_1 + (1-c) m_2$ is a convex combination of $m_1$ and $m_2$.
We have the following equalities of distribution functions: 
\begin{align}
\mu_{c p_1 + (1-c) p_2} &= c \mu_{p_1} + (1-c) \mu_{p_2} \\
\xi_{c m_1 + (1-c) m_2} &= c \xi_{m_1} + (1-c) \xi_{m_2}
\end{align}
\end{pos}

\begin{example*}
In any physical experiment, the convex combination of two operations can be seen as a probabilistic sample of these operations. 
In the Hilbert space formalism for quantum mechanics, it is the convex combination of a set of density matrices or observables, respectively.
\end{example*}

The contextuality examples arise from the fact that POVM's and mixed density matrices can be written as different convex combinations of projectors and pure density matrices, respectively.
It is not possible that the convex decompositions are preserved at the ontological level for all combinations at once. We will discuss the two examples in detail in section~\ref{sec:conv}.

\subsection*{Outcome-Determinism}
An ontological representation is {\bf outcome-deterministic} if $\xi_m(\lambda)(k) \in \{0,1\}$ for all $m\in M, k \in O$ and $ \lambda \in \Omega$. In~\cite{spekkens} it was shown that every non-contextual ontological representation for PVM's in quantum mechanics must be outcome-deterministic.  We generalise this argument below for operational theories.
To characterize a class of operational theories in which outcome-determinism can be justified, we generalise the notion of sharp measurements and the maximally mixed state in quantum mechanics to the operational notions below.

\begin{definition}
A measurement procedure $m$ with outcome set $O^m$ is {\bf perfectly predictable} if for all $k \in O^m$ there exists a preparation $p_k$, such that $d_{p_k,m}(k') = \delta_{k,k'}$.
\end{definition}

\begin{definition}
A preparation $p^{mix}$ is {\bf maximally mixed} if the following two conditions hold.
\begin{enumerate}
\item For every preparation procedure $p'$, $p^{mix}$ is statistically equivalent to some preparation procedure which is a convex combination of preparations containing $p'$. 
\item For every perfectly predictable measurement procedure $m$, $p^{mix}$ is statistically equivalent to some measurement procedure which is a convex combination of $p_k$ for $k \in O^m$.  
\end{enumerate}
\end{definition}

\begin{example*}[Quantum Mechanics]
In quantum theory, PVM's are perfectly predictable. In the Hilbert space model for quantum theory of a fixed dimension $d$,
any linear decomposition of the maximally mixed state $I/2^d$ is a maximally mixed preparation.
\end{example*}

\begin{lemma}\label{lem:od}
Let $m$ be a perfectly predictable measurement procedure in an operational theory that satisfies assumption~\ref{pos:conv} and contains a maximally mixed preparation. If the distribution function $\xi_m$ is part of a preparation non-contextual ontological representation, it is outcome-deterministic.
\end{lemma}

\begin{proof}
Let $\Omega_p := \{\lambda \in \Omega | \mu_p(\lambda) > 0\}$ be the support of $\mu_p$. 
Perfect predictability implies that for each measurement $m$ and outcome $k$ of $m$, there exists a preparation $p_k$, such that \newline $d_{p_{k},m}(k') = \sum \xi_{m}(\lambda)(k')\mu_{p_{k}}(\lambda)= \delta_{k,k'}$ for all $k' \in O^m$.
If $k = k'$, this means that $\xi_{m}(\lambda)(k) = 1$ for all $\lambda \in \Omega_{p_k}$, since $\mu_{p_{k}}$ is a probability distribution and $\xi_m(\lambda)(k) \leq 1$. For $k \neq k'$, this implies that $\xi_{m}(\lambda)(k) = 0$ for $\lambda \in \Omega_{p_{k'}}$. As a result, $\xi_m(\lambda)(k) \in \{0,1\}$ for $\lambda \in \cup_k \Omega_{p_{k}}$.

It is left to prove that $\Omega = \cup_k \Omega_{p_k}$. By preparation non-contextuality, there is one probability distribution $\mu_{p^{mix}}$ for all preparations that are statistically equivalent to the maximally mixed preparation $p^{mix}$.  We use this to show that $\Omega = \Omega_{p^{mix}}$. It is clear that $\Omega_{p^{mix}} \subset \Omega$, so we need to show that for all $\lambda \in \Omega$,  it is true that $\lambda \in\Omega_{p^{mix}}$. We may assume that for every element $\lambda \in \Omega$ there exists a preparation $p$, such that $\lambda \subset \Omega_{p}$. Otherwise, it is impossible to ever obtain the value $\lambda$. By the first condition of the maximally mixed preparation and assumption~\ref{pos:conv},  we know that $\lambda \subset \Omega_{p^{mix}}$. It follows that $\Omega = \Omega_{p^{mix}}$. By the second condition and assumption~\ref{pos:conv}, we know that $\mu_{p^{mix}} = \mu_{\sum c_k p_k}$. It follows from the definition of the support, that $ \Omega_{p^{mix}} = \Omega_{\sum c_k p_k} =  \cup \Omega_{c_k {p_k}} = \cup { \Omega_{p_k}}$, since all $\mu_{p_k}$ are strictly positive functions.  It follows that $\Omega = \cup \Omega_{p_k}$, hence $\xi_m$ is outcome-deterministic.
\end{proof}
\subsection*{Beyond Outcome-Determinism}

Another setup which imposes additional restrictions on non-contextual ontological representations is that of joint measurements. We recall the definition given in ~\cite{liangspekkenswiseman}.

\begin{definition}\label{def:jointm}
A set of $N$ measurements $\{m_1, m_2,..., m_N\}$ is {\bf jointly measurable} if there exists a measurement $m$ with the following features: 
\\(i) The outcome set of $m$ is the Cartesian product of the outcome sets of $m_1,...,m_N$
\\(ii) Let $S$ be a subset of the index set $\{1,...,n\}$. The outcome distributions for every joint measurement of any subset $\{m_s|s \in S\} \subset \{m_1,...,m_N\}$ is recovered as the marginal of the outcome distribution of $m$ for all preparations $p \in P$. Denoting a joint measurement of the subset $S$ by $m_S$ with a corresponding section $k_S \in O^{m_S}$, the condition can be expressed as 
\begin{equation}
 \forall S, \forall p: d_{p,m_S}(k) = \sum_{k\in O^m: \pi_S(k)=k_S} d_{p,m}(k).
\end{equation}
Here, $\pi_S$ is the projection function on the subset $\mathcal{E}(m_S) \subset \mathcal{E}(m)$.
\\(iii) The composition of functions $\pi_S \circ m$ corresponds to a measurement in the operational theory for each subset $S\subset \{1,...,N\}$.
\end{definition}

Condition (iii) was not given in the original definition. As we will see,  with this additional condition, equivalence-based non-contextuality implies that any ontological representation of an operational theory with joint measurements is parameter independent. An ontological representation is called {\bf parameter independent} if for each measurement the effect on the ontological states is independent of any other measurement performed simultaneously. We can restrict the distribution function of a joint measurement $m$ to a subset $m'$ by the restriction function ${\xi_m}|_{m'}$, which is defined as ${\xi_m}|_{m'}(k')(\lambda) : = \sum_{k: \pi_{m'}(k) = k'} \xi_{m}(k)(\lambda)$, where $\pi_{m'}: O^m \rightarrow O^{m'}$ projects the outcomes of $m$ to the set of outcomes of $m'$. Parameter independence means that for two joint measurements $m,n$, the equality ${\xi_m}|_{m \cap n} = {\xi_n}|_{m \cap n}$ holds.  
\begin{lemma}\label{lem:nosig}
Any measurement non-contextual representation of an operational theory $(P, M, D, O)$ is parameter independent. That is, for all joint measurements $\{m,n\}$, we have the following equalities:
\begin{equation} \xi_{m|m\cap n}^{k}(\lambda) = \xi_{m\cap n}^k(\lambda) = \xi_{n|m\cap n}^{k}(\lambda)\end{equation}
\end{lemma}

\begin{proof}
Let $m=\{m_1,...,m_N\}$ be a jointly measurable set of measurement procedures of $M$, let $p$ be a preparation procedure, let $K_s$ be the set $\{k \in O^m | \pi_s(k) = k_s\}$, for some $s \in \{1,...,N\}, k_s \in O^{m_s}$. 

By joint measurability and basic probability theory, we have the following sequence of equalities on the operational level:

\begin{align*}
d_{m_s,p}(k_s) &= \sum_{k \in K_s} d_{m,p}(k) \\
& = d_{m,p}(K_s)\\
&= d_{\pi_s \circ m,p}(k_s)
\end{align*}

By condition $(iii)$ of Definition~\ref{def:jointm}, $\pi_S \circ m$ is a well-defined measurement. It then follows from measurement non-contextuality that $\xi_{m_s}(k_s)= \xi_{\pi_s \circ m}(k_s)$, which implies

\begin{equation}
\xi_{m_S}(k_s) = \sum_{k| \pi_s(k)=k_s} \xi_m(k) =: \xi_{m|m_s}(k_s)
\end{equation}
\end{proof}

A stronger restriction on ontological representations is factorizability. We call an ontological model {\bf factorizable} when for joint measurements $m=(m_1,...,m_n)$, we can write \newline $\xi_m(\lambda)(o) = \prod_{i=1,...,n} \xi_{m_i}(\lambda)(\pi_i(o))$. It was shown in Theorem 6 of~\cite{liangspekkenswiseman} that any ontological model which is outcome-deterministic and measurement non-contextual, is factorizable. In fact, we will show in Section~\ref{sec:uni} that an operational theory admits a factorizable non-contextual ontological representation iff it admits a deterministic non-contextual ontological representation. This is a generalisation of Fine's theorem~\cite{fine} for the equivalence-based notion of contextuality.

\subsection*{Bell's Scenario}\label{sec:Bell} 
Consider the following experiment where two parties can each choose from two different measurements with outcome set $\{1,-1\}$: $a$ and $a'$ for the first party; $b$ and $b'$ for the second party. The outcome statistics of each possible combination of measurements after a preparation $p$ is organised in the table below. Each entry $a_{i,j}$ of the table represents the probability of obtaining outcome $i$ for measurement $j$.

\begin{align}
\begin{tabular}{ c | c  c c c }
  \hspace{.1cm} & (1,1) & (-1,1) & (1,-1) & (-1,-1)\\
  \hline			
  $(a,b)$ & 1/2 & 0 & 0 & 1/2\\
  $(a',b)$ & 3/8 & 1/8 & 1/8 & 3/8\\
  $(a,b')$ & 3/8 & 1/8 & 1/8 & 3/8\\
  $(a',b')$ & 1/8 & 3/8 & 3/8 & 1/8\\
\end{tabular}
\end{align}

This setup can be realised in the operational theory given by the Hilbert space formalism of quantum mechanics. Consider the quantum states

\begin{align*}
&\phi_{a_+}, \phi_{b_+} &= &(\ket{0} + \ket{1}) / \sqrt{2})  \hspace{1cm}
&\phi_{a_-}, \phi_{b_-} &= &(\ket{0} - \ket{1})\ / \sqrt{2})\\
&\phi_{a'_+}, \phi_{b'_+} &= &(\ket{0} + e^{\pi/3 i}\ket{1}) / \sqrt{2}) \hspace{1cm}
&\phi_{a'_-}, \phi_{b'_-} &= &(\ket{0} + e^{(\pi/3 + \pi) i}\ket{1}) / \sqrt{2})\\
&\phi_{GHZ} &= &(\ket{00}+\ket{11})/ \sqrt{2} 
\end{align*}

and observables $P_{x_+}$ and $P_{x_-}$ that project onto the state $\phi_{x_+}$  and $\phi_{x_-}$, respectively; as well as the measurement $x:= (P_{x_+}, P_{x_-})$ for $x \in \{a,b,a',b'\}$. 

We are interested in measurement procedures $(a,b)$, $(a',b)$, $(a,b')$, and $(a',b')$ corresponding to the PVM's $a \otimes b, a' \otimes b, a \otimes b'$ and $a' \otimes b'$, respectively, as well as the restriction to each of the components for each of the measurements, which are obtained by taking the partial trace.

We have seen before that the Hilbert space formalism of fixed dimension has a well-defined notion of convex combinations and contains a maximally mixed preparation. Furthermore, all measurements are perfectly predictable as they are PVM's. The tuples are joint measurements of their restrictions, which fall into four equivalence classes $a,a',b,b'$. Since any non-contextual ontological representation is in particular preparation non-contextual, it follows that any such representation must be outcome-deterministic under assumption~\ref{pos:conv}, hence factorizable.  As a result, every ontological state can be associated with a function $\{a,a',b,b'\} \rightarrow \{0,1\}$ from the set of equivalence classes of elementary measurements to the outcome set. This morphism maps each measurement to the outcome that occurs with probability 1 when the system is in ontological state $\lambda$. In other words, we can identify each ontological value $\lambda$ with the outcome $(\alpha, \beta, \alpha', \beta') \in \{0,1\}^4$. We obtain the presupposed probability that $a=\alpha, b=\beta, a'=\alpha', b'=\beta'$ after the preparation of $\phi_{GHZ}$ by integrating  according to equation~\ref{def:or} over all ontological values that correspond to this outcome. We denote this probability by $p_{\alpha \beta \alpha' \beta'}$. These probabilities should sum up to the values given in the table. The entries $a_{1,1}, a_{2,2}, a_{3,3}$, and $a_{1,4}$, give us the 4 equations below.

\begin{align*}
&a_{1,1}: p_{0000} + p_{0010} + p_{0001} + p_{0011} = 1/2 \hspace{.5 cm}
&a_{2,2}: p_{0010} + p_{1010} + p_{0011} + p_{1011} = 1/8\\
&a_{3,3}: p_{0001} + p_{0101} + p_{0011} + p_{0111} = 1/8 \hspace{.5 cm}
&a_{1,4}: p_{0000} + p_{0100} + p_{1000} + p_{1100} = 1/8
\end{align*}
The left-hand-side of the sum of $a_{2,2}, a_{3,3}$ and $a_{1,4}$ should be greater than $1/2$, since it contains all summands of $a_{1.1}$. However, the right-hand side of these equations sums to $3/8$. As a result, the equations cannot be satisfied. As a consequence of this contradiction, a non-contextual ontological representation cannot exist.

\subsection*{Mermin's All Versus Nothing Argument}\label{sec:Mermin}
For the next example, we again consider the operational theory defined by the Hilbert space formalism of quantum mechanics. Suppose that we are given a GHZ state $\phi_{GHZ} = (\ket{000} +  \ket{111})/\sqrt{2}$ and we may perform Pauli $X$ or $Y$ measurements on each of its components. One can verify that for a choice of joint measurements, the following equalities hold with certainty. The right-hand-side of the equalities are given by the product of the outcomes of the three individual measurements: 

\begin{align*}
&X_1Y_2Y_3 &= -1 \hspace{1cm} 
&Y_1Y_2X_3 &= -1 \hspace{1cm}
&Y_1X_2Y_3 &= -1 \hspace{1cm}
&X_1X_2X_3 &= 1
\end{align*}

As in the previous example, the measurements are perfectly predictable, and the operational theory contains a maximally mixed state. Furthermore, the triples are joint measurements of their restrictions to the three different components, given by the partial trace. These restrictions fall into the equivalence classes $X_1, Y_1, X_2, Y_2, X_3, Y_3$.
We will show that no ontic state $\lambda$ in a non-contextual representation allows for probability distributions $\mu_{-}(\lambda)$ that are consistent with this scenario.
Suppose that there exists a non-contextual ontological representation for this operational theory. In particular, this representation is preparation non-contextual. Perfect predictability and the maximally mixed preparation imply outcome determinism and factorizability.  Hence, $\mu_{X_1Y_2Y_3} = \mu_{X_1} \mu_{Y_2} \mu_{Y_3}$, and similarly for the other joint measurements.
Given any ontological state $\lambda$ of such representation, we can identify each of the measurements $X_i, Y_i$ with the outcome that occurs with certainty for $\mu_{X_i}(\lambda)$ and $\mu_{Y_i}(\lambda)$, respectively. By factorizability, the outcomes of the joint measurements correspond to the product of the outcomes of the three components. In other words, $\mu$ assigns $-1$ or $1$ to each $X_i, Y_i$, in a way that the equalities above are satisfied. It is easy to see that this is impossible: The product of the expressions on the left-hand-side must equal 1, since every measurement occurs twice, while the product of the right-hand-sides equals -1.

\section{The Sheaf Approach}\label{sec:emp}
We recall the sheaf-theoretic approach to contextuality and non-locality, which was introduced by Abramsky and Brandenburger in~\cite{abramskybrandenburger} and ~\cite{abramskybrandenburger2}. We consider the version of the formalism given in~\cite{abramskysadrzadeh} and extend it to incorporate preparations. As we will show later, an extension of this version is applicable to any operational theory.

 Sheaves are a mathematical tool for describing how local data can be combined to obtain global information about a system. In this setting, a {\bf system type} consists of a discrete set $X$ of measurement labels, together with a {\bf measurement cover} $\cM = \{C_i\}_{i \in I}$. This is an antichain of subsets $C_i \subset X$, such that $\cup_{i \in I} C_i = X$. This means that for $C, C' \in \mathcal{M}$, we have the implication $C \subset C' \implies C=C'$. The measurement cover $\cM$ represents the maximal sets of measurements that can be performed jointly. We write $\downarrow \cM_A$ for the simplicial complex generated by $\cM$. 

We shall fix a set $O$ of outcomes, which is the union of the sets of possible outcomes for each of the measurements in $X$. For each set of measurements $U\subset X$, a {\bf section over $U$} is a function $U \rightarrow O$. We write $O^U$ for the set of sections over $U$. The assignment $U \mapsto O^U$ defines a sheaf over the discrete topological space $\cE: \mathcal{P}(X) \rightarrow Set$, which we call the {\bf sheaf of events}.  The restriction function, which is the remaining part of the data defining this sheaf, is given below.
\[ \rho^U_{U'}:= \cE(U \subset U'): O^{U'} \mapsto O^U :: s \rightarrow s|_U\]
We call elements of $\cE(X)$ {\bf global sections} of measurement outcomes. Each global section consists of an assignment of an outcome to each of the measurements.

For any commutative semiring R and set X, an $R$-distribution $d$ on $X$ is a map $d : X \rightarrow R$ of finite support, such that 
 
\[ \sum_{x \in X} d(x) = 1\]

We write $\mathcal{D}_R(X)$ for the set of $R$-distributions on $X$. For a function of sets $f: X \rightarrow Y$, we define
\[ \mathcal{D}_R(f):\mathcal{D}_R(X) \rightarrow \mathcal{D}_R(Y) :: d \mapsto [ y \mapsto \sum_{f(x)=y}d(x)]\]
It is easy to see that $\mathcal{D}_R$ is functorial. Hence, we can compose $\mathcal{E}$ with $\mathcal{D}_R$ to obtain a presheaf $\mathcal{D}_R\mathcal{E}: \mathcal{P}(X)^{op} \rightarrow Set$, which maps each set of measurements to the set of $R$-distributions over their sections. When $R$ is the ring of non-negative reals $\mathbb{R}_+$,  $\mathcal{D}_R\mathcal{E}(m)$ corresponds to probability distributions over the outcomes of $m$; if $R$ is the ring of booleans $\mathbb{B}$, it represents the possibility of outcomes of $m$. 
 
The approach can be generalised to a presheaf over a small, thin category 
 $D_{\R}\cE: {\bf C} \rightarrow Set$, as in~\cite{abramskysadrzadeh}.  A category is called thin when for each two objects $A,B$ and each two morphisms $f,g: A\rightarrow B$, we have the equality $f=g$. This is the categorical way to characterise a preorder. The order relation is given by a notion of joint measurement.
Depending on the interpretation of an empirical theory, one could adopt different notions of joint measurement. In this paper, we will use the notion of joint measurability given in Definition~\ref{def:jointm}. 
We recover the set of measurement labels $X$ as the set of objects that are jointly measurable sets of at most one element.  These are all objects $A$, such that there is an arrow $A \rightarrow B$ to each object $B$, for which there exists an arrow $B \rightarrow A$. The measurement cover $\mathcal{M}$ corresponds to the maximal jointly measurable sets. This cover contains those objects $A$, such that there is an arrow $B \rightarrow A$ from every object $B$, for which there exists an arrow $A \rightarrow B$. Note that when the thin category is a poset, we obtain the usual notions of measurement labels and a measurement cover.

A {\bf state} for a system type ${\bf C}$ determines a distribution $\sigma_C \in \mathcal{D}_R\mathcal{E}(C)$ for each measurement context $C \in \mathcal{M}$. A state is called {\bf no-signalling} when for all $C,C'\in \mathcal{M}$
\[ \sigma_C|_{C \cap C'} = \sigma_{C'}|_{C \cap C'}\]
When a state $\sigma$ is no-signalling, the restriction $\sigma_{C|_m}$ for $m \in Ob({\bf C})$ corresponds to the same probability distribution over the outcomes of $m$ for each $C \in \cM$. We will denote this distribution by $\sigma_m$. There may be several states corresponding to the same distribution. These are {\bf statistically equivalent states}. We call the tuple $({\bf C}, S, O)$, where $S$ is a collection of states for a system type ${\bf C}$, an {\bf empirical theory}. An empirical theory is no-signalling if all states in $S$ are no-signalling.
A {\bf global section} for the presheaf $\mathcal{D}_R\mathcal{E}$ is given by an indexed set $d = \{d^{\sigma}\}_{\sigma \in S}$, such that for each $\sigma \in S$, $d^{\sigma} \in \mathcal{D}_R\mathcal{E}(X)$ and we have the following equality
\begin{center}
$d^{\sigma}_{m} := d^{\sigma}_X|_m = \sigma_m$
\end{center}

We define {\bf contextuality of an empirical theory} as the non-existence of a global section for the presheaf $\mathcal{D}_R\mathcal{E}$.

\begin{example*}[Quantum Mechanics]
The Hilbert space model for quantum mechanics gives rise to an empirical theory. The measurement labels correspond to POVM's; the measurement cover consists of maximal sets of joint measurement as in Definition~\ref{def:jointm}; and states are given by density matrices, which determine the corresponding families of probability distributions according to the Born rule. 
\end{example*}

\subsection*{Convexity in Empirical Theories}\label{sec:conv}
Many of the classic contextuality results about quantum mechanics, including Kochen-Specker scenarios~\cite{KS}, Hardy's paradox~\cite{hardy}, Bell's scenario~\cite{bell} and Mermin's all versus nothing argument~\cite{mermin} can be derived from the sheaf approach, as shown in~\cite{abramskybrandenburger} and \cite{asklm}.
In this section, we show that the same is true for the contextuality of preparations and unsharp measurements~\cite{spekkens}. Therefore, these contextuality arguments are independent from the chosen notion of contextuality.
We will demonstrate that the arguments are a direct consequence of assumption~\ref{pos:conv}.

As assumption~\ref{pos:conv} is formulated in terms of operational theories and ontological representations, we introduce an analogue for empirical theories. In this setting, global sections can be seen as the counterpart of the non-contextual ontological representations. We explain this in Section~\ref{sec:canhidvar}.

\begin{pos}[Preservation of convexity]\label{pos:convemp}
Let $d$ be a global section for an empirical theory. Let $\sigma^{p_1}, \sigma^{p_2}, \sigma^{c_1 \cdot p_1 + c_2 \cdot p_2} \in S$ be states, where $\sigma^{c_1 \cdot p_1 + c_2 \cdot p_2}:= c_1 \cdot \sigma^{p_1} + c_2 \cdot \sigma^{p_2}$ is the formal convex combination of $\sigma^{p_1}$ and $\sigma^{p_2}$. Let $m_1$, $m_2, c_1 \cdot m_1 + c_2 \cdot m_2 \in X$ be measurement labels, where $c_1 \cdot m_1 + c_2 \cdot m_2$ is the convex combination of $m_1$ and $m_2$. The equalities below hold.

\begin{align}
d^{\sigma_{c_1 \cdot p_1 + c_2 \cdot p_2}} &= c_1 \cdot d^{\sigma_{p_1}} + c_2 \cdot d^{\sigma_{p_2}} \\
d^{\sigma}_{c_1 \cdot m_1 + c_2 \cdot m_2} &= c_1 \cdot d^{\sigma}_{m_1} + c_2 \cdot d^{\sigma}_{m_2} \hspace{2cm} \forall \sigma \in S
\end{align}
\end{pos}

Note that $\sigma^{c_1 \cdot p_1 + c_2 \cdot p_2}$ is a formal convex combinations of states $\sigma^{p_1}$ and $\sigma^{p_2}$, not the convex combinations of their corresponding probability distributions.

\begin{lemma}\label{lem:convemp}
Let $({\bf C}, S)$ be an empirical model with a global section $d$. Assumption \ref{pos:convemp} implies that convexity in an empirical model is preserved by the probability distributions defined by states.

\begin{align*}
\sigma_C^{c_1p_1 + c_2p_2} &= c_1 \cdot \sigma_C^{p_1} + c_2 \cdot \sigma_C^{p_2} \hspace{1cm} \forall C \in \cM \\
\sigma_{c_1 m_1 + c_2 m_2} &= c_1 \cdot \sigma_{m_1} + c_2 \cdot \sigma_{m_2} 
\end{align*}
\end{lemma}

\begin{proof}
This follows immediately from the definition of a global section.
\end{proof}

The probability distributions in the empirical model of quantum mechanics do not preserve the convexity of the states and measurements. This is due to the fact that  the Born rule does not preserve convexity of POVM's or density matrices. 
We show this below by considering different decompositions of the maximally mixed state and of the maximally mixed POVM; hence, we prove sheaf-theoretic contextuality for preparations and unsharp measurements in 2-dimensional quantum mechanics.

\subsection*{Contextuality for Preparations}\label{sec:conv}

Consider the following set of states in the Hilbert space formalism of quantum mechanics: 

\begin{align*}
&\psi_a &= &(1,0)\hspace{1cm}
&\psi_b &= &(1/2, \sqrt{3}/2)\hspace{1cm}
&\psi_c &= &(1/2, -\sqrt{3}/2)\\
&\psi_A &= &(0,1)\hspace{1cm}
&\psi_B &= &(\sqrt{3}/2, -1/2)\hspace{1cm}
&\psi_C &= &(\sqrt{3}/2, 1/2)
\end{align*} 

We define the empirical theory below, where each $P_x$ is the measurement label that corresponds to the projection onto the quantum state $\phi_x$. Furthermore, $\sigma^{\phi_x}$ is the state in the empirical model that corresponds to the quantum state $\phi_x$. The state $\sigma^{\phi_{mix}}$ corresponds to the maximally mixed state $\phi_{mix} = (1/\sqrt{2}, 1/\sqrt{2})$.

\begin{align*}
&X = \{ P_a, P_A, P_b, P_B, P_c, P_C\} \hspace{.5cm}
&\mathcal{M} = \{\{P_a, P_A\}, \{P_b, P_B\}, \{P_c, P_C\}\}\\
&S = \{ \sigma^{\phi_a}, \sigma^{\phi_A}, \sigma^{\phi_b}, \sigma^{\phi_B}, \sigma^{\phi_c}, \sigma^{\phi_C}, \sigma^{\phi_{mix}}\}\hspace{.5cm}
&O =\{0,1\}
\end{align*} 

The outcome set $O$ indicates if the outcome corresponding to the projector of the POVM element occurs (1), or if it does not (0).  For instance, for the section $s: P_a \mapsto 0$, $\sigma^{\phi_a}_{P_a}(s) = 0$, $\sigma^{\phi_b}_{P_a}(s) = \frac{1}{4}$, and $\sigma^{\phi_c}_{P_a}(s) = \frac{1}{4}$. 

The linear combination of density matrices $\frac{1}{3} \phi_a + \frac{1}{3} \phi_b + \frac{1}{3} \phi_c$ is equal to the maximally mixed state $\phi_{mix}$ for any observable $P_x$. Suppose that there exists a global section $d$, by assumption~\ref{pos:convemp} and Lemma~\ref{lem:convemp}, this gives us the following equality.

\begin{align*}
\frac{1}{2}
= \frac{1}{3} \sigma^{\phi_a}_{P_a} + \frac{1}{3} \sigma^{\phi_b}_{P_a} + \frac{1}{3} \sigma^{\phi_c}_{P_a}
\end{align*}

It is easy to see that this cannot hold for any section. Working out the outcome probabilities for $P_a \mapsto 0$ gives us the contradiction below. 

\begin{align*}
\frac{1}{2}&= \frac{1}{3} \cdot 0 + \frac{1}{3} \cdot \frac{1}{4} + \frac{1}{3} \cdot \frac{1}{4}
\end{align*}



\subsection*{Contextuality for Unsharp Measurements} 
Consider the following empirical theory 

\begin{align*}
&X = \{ P_a, P_A, P_b, P_B, P_c, P_C, P_{abc}, P_{ABC}\} \hspace{.5cm}\\
&\mathcal{M} = \{\{P_a, P_A\}, \{P_b, P_B\}, \{P_c, P_C\}, \{P_{abc}, P_{ABC}\}\}\\
&S = \{ \sigma^a, \sigma^A, \sigma^b, \sigma^B, \sigma^c, \sigma^C\}\hspace{.5cm}\\
&O =\{0,1\}
\end{align*} 

The measurement label $P_{abc}$ is the convex combination $\frac{1}{3} P_a + \frac{1}{3} P_b + \frac{1}{3} P_c$, $P_{ABC}$ is defined similarly, and the other elements are as defined in section~\ref{sec:conv}. In quantum theory, this is the uniform sample over the respective projectors. This gives us the measurement context $\{P_{abc},P_{ABC}\} = \{\frac{1}{2}, \frac{1}{2}\}$. 
Suppose that this scenario has a global section. By assumption~\ref{pos:convemp} and Lemma~\ref{lem:convemp}, we have the  equalities below for any state~$\sigma$.

\begin{align}
&\frac{1}{2} = \sigma_{P_{abc}} =  1/3 \sigma_{P_a}+ 1/3 \sigma_{P_b} + 1/3 \sigma_{P_c}\\
&\frac{1}{2} = \sigma_{P_{ABC}} = 1/3 \sigma_{P_A} + 1/3 \sigma{P_B} + 1/3 \sigma{P_C}
\end{align}

It is easy to see that this does not hold for the given states. For example, if we take the state $\sigma^a$, the convexity condition together with the Born rule give us $\sigma^a_{p_{abc}}(1) = \sigma^a_{P_{ABC}}(0) = \frac{3 + \sqrt{3}}{6}$ and $\sigma^a_{P_{ABC}}(1) = \sigma^a_{P_{abc}}(0) = \frac{1 + \sqrt{3}}{6}$. This contradicts the outcome statistics of $\{p_{abc}, p_{ABC}\}$, which assign equal probability to each outcome for any measurement.  

\section{Unifying Approaches}\label{sec:uni}
In this section, we will explore the relation between empirical theories, operational theories, and ontological representations.  We establish a link between general equivalence-based contextuality and sheaf-theoretic contextuality. In other words, between theories that admit no ontological representation that is preparation non-contextual and measurement non-contextual according to equivalence-based formalism and theories that admit no global section $d \in \mathcal{D}_{\mathbb{R}}(\mathcal{E}(X))$ in the sheaf-theoretic formalism, respectively.

Any no-signalling empirical theory $A=({\bf C}_A, S_A, O_A)$ corresponds to an operational theory $Op(A)= (P_{Op(A)}, M_{Op(A)}, D_{Op(A)}, O_A)$ in the sense that the two theories describe the same experimental setting. The elements of the operational theory are defined below.
\begin{align*}
P_{Op(A)}&:= S_{A} \\
M_{Op(A)}&:= Ob({\bf C}_A) \\
d_{m, \sigma}(k)&:=\sigma_{m}(s) \hspace{2cm} \mbox{ for } d_{m,\sigma} \in D_{Op(A)} \mbox{ and } s(m)=k
\end{align*}
Conversely, every set of preparation procedures $P$ together with the set distributions $D$ give rise to a set of states $S$; every set of measurement procedures $M$ gives rise to a set of measurement labels $X$ consisting of 1-element joint measurements; and the preorder defined by joint measurements in $M$ gives rise to the thin category ${\bf C}_A$.

\begin{remark}\label{rem:sign}
Signalling empirical theories cannot be described as an operational theory. The reason is that while the 'same' measurement can have different outcome statistics in the empirical theory, depending on the context, this is not possible in an operational theory. A way to get around this is by treating restrictions of a context to a measurement as elementary measurements. 

\end{remark}
We will show that every sheaf-theoretic non-contextual empirical theory $A$ gives rise to a non-contextual ontological representation for $Op(A)$. We will call this a canonical ontological representation for the operational theory. Finally, we prove the theorem below, which connects sheaf-theoretic contextuality to equivalence-based contextuality. This theorem generalises the result in~\cite{abramskybrandenburger}, as well as Fine's Theorem~\cite{fine}, to the more general setup of sheaf-theoretic contextuality in empirical theories and equivalence-based contextuality in operational theories.

\begin{theorem}\label{thm:eq}
The following statements are equivalent for any no-signalling empirical theory $A$ and its corresponding operational theory $Op(A)$ 
\begin{enumerate}
\item The empirical theory $A$ admits a global section 
\item The operational theory $Op(A)$ admits a canonical non-contextual ontological representation
\item  The operational theory $Op(A)$ admits a factorizable non-contextual ontological representation
\end{enumerate}

\end{theorem}

\subsection*{A Canonical Ontological Representation}\label{sec:canhidvar}
As a warm-up, we recall the canonical ontological representation for empirical models with a global section $d$, which was introduced in~\cite{abramskybrandenburger}. An empirical model corresponds to an empirical theory with only one state. The ontological states are given by the global sections of outcomes, the distributions $\mu$ correspond to the global section of distribution functions and $\xi_{m}(s)(k)$ indicates whether $s$ assigns the outcome $k$ to the measurement $m$.

\begin{align*}
\Omega = \mathcal{E}(X) \hspace{1cm} \mu_{\sigma}(s)=d(s) \hspace{1cm} \xi_{m}(s)(k) = \delta_{s|_{m}(m),k}
\end{align*}

It is easy to see that this ontological representation is generally not non-contextual. The sections may assign different outcomes to statistically equivalent measurements. Suppose that $s$ is a section of measurement outcomes such that $s|_m \neq s|_n$ for $m \sim n$, then $\xi_m(s)(k) \neq \xi_n(s)(k)$. 
 To get around this, we will prove that whenever a global section exists, we can find another global section that depends on equivalence classes of measurements only. It is not hard to see that the same holds for states.

\subsection*{Statistical Equivalence in Empirical Theories}\label{sec:stateq}
We call two states $\sigma, \sigma' \in S$ and two measurement labels $m,m' \in Ob({\bf C})$ {\bf statistically equivalent} when $\sigma_m = \sigma'_m$ for all $m \in Ob({\bf C})$ and $\sigma_m = \sigma_{m'}$ for all $\sigma \in S_A$, respectively. In that case we write $\sigma \sim \sigma'$ and $m \sim m'$. 

Let $A=({\bf C}_A, S_A)$ be an empirical theory.  We construct a new empirical theory $\tilde{A}:=({\bf C}_A/{\sim},  \tilde{S}_A)$ by quotienting the objects of ${\bf C}$ by the equivalence relation. The new category ${\bf C}_A/\sim$ contains an arrow between two equivalence classes if there exists an arrow between two representatives of the classes. 
It is instructive to unfold the structure of this new empirical theory. For each object $[C]$ of ${\bf C}_A/\sim$ the new set of sections $\mathcal{E}([C])$ contains a (not necessarily unique) section $\tilde{s}$ for each $s \in \cE(C)$. This section is defined as $\tilde{s}([C]) := s(C)$. 
The states in $\tilde{S}_A:= \{\tilde{\sigma}\}_{\sigma \in S_A}$, are defined as $\tilde{\sigma}_{[C]}(\tilde{s}):= \sigma_C(s)$.
The set $\tilde{S}_A$ is well-defined, because $[C]=[D]$ if and only if $\sigma_C= \sigma_D$ for each $\sigma \in S_A$.

\begin{lemma}\label{lem:gs}
Any empirical theory $A$ admits a global section iff it admits a global section that only depends on equivalence classes of measurements of $A$.
\end{lemma}

We will prove this Lemma formally in Section~\ref{sec:catiso}. Intuitively, it can be understood as follows: Any global section $d$ of $S_A$ can be restricted to a global section over a subcategory of ${\bf C}_A$ of representatives of ${\bf C}_A/\sim$. This restriction defines a global section for $\tilde{S}_A$. Conversely, any global section $\tilde{d}$ of $\tilde{A}$ defines a global section $d$ for $A$, defined as $d(s):= \tilde{d}({\tilde{s}})$ when $s$ assigns the same value to all elements of an equivalence class, and $d(s):= 0$ otherwise. 

We have shown how to deal with equivalence on the level of measurements. However, individual outcomes of measurements can be statistically equivalent, even when the measurements as a whole are not. This means that for some $s \in O^m$ and $s' \in O^{m'}$, $\sigma_m(s) = \sigma_{m'}(s')$ for all $\sigma \in S_A$. To eliminate this last form of statistical redundancy, we rewrite any such system type $A$ as a system type $A'$ with outcome set $\{0,1\}$.  The measurement labels of $A'$ are given by the individual observables in each measurement. We denote each observable by a tuple $(m,k)$ of a measurement and an outcome, so $X_{A'} :=  \{(m,k)\}_{m \in X_A, k \in O}$.
The measurement cover is given by the sets of observables that form a measurement in the original cover: $\mathcal{M}_{A'} := \{\{(m,k)|k \in O, m \in C\}_{C \in \mathcal{M}_A}\}$. The outcomes $0$ and $1$ indicate whether the outcome corresponding to the observable is observed, hence $S_{A'} := \{\sigma'| \sigma'_{(m,k)}(1) = \sigma_m(k)\}$. The support $\mathcal{E}(m)$ of each measurement $m$ consists of those sections where exactly one observable in each measurement is assigned a $1$, and all others are assigned a $0$.  


 Note that the model $A$ has a global section iff $A'$ has a global section under the given restrictions. As a consequence of Lemma~\ref{lem:gs}, $A$ has a global section iff $\tilde{A'}$ has a global section. Hence, $A$ contains a global section induced by a global section $d$ for $\tilde{A'}$, which is only defined on equivalence classes.  
 
\subsection*{Non-contextual Canonical Ontological Representations}
We can now define a canonical ontological representation that preserves non-contextuality. 
Let $A$ be an empirical theory with a global section $d_\sigma$ for each state $\sigma \in S_A$, which only depends on the equivalence classes of the preparations. We make use of the minimal empirical theory $\tilde{A'}$ and its induced global sections $\tilde{d}_{\tilde{\sigma}}$ to define the canonical ontological representation $R(A)=(\Omega^{NC}_A, \{\mu^{NC}_{\sigma}\}_{\sigma \in S_A}, \{\xi^{NC}_m\}_{m \in \downarrow \mathcal{M}})$:

\begin{align*}
\Omega_{R(A)} &:= \mathcal{E}(X({\bf C/\sim})), \hspace{.5cm}
&\mu^{NC}_{\sigma}(s) &:= \tilde{d}_{\tilde{\sigma}}(\tilde{s}), \hspace{.5cm}
&\xi^{NC}_m(s)(k) &:= \delta_{\tilde{s}([m]),[k]}
\end{align*}

Note that $\xi_{m}^{NC}(s)(k)$ is only defined when $s$ is a section over $n$, so when this is not the case, we will take $\xi_m^{NC}(s)(k)$ to be $0$. This representation generates the required outcome statistics, as shown below.

\begin{align*}
\sum_{s \in \Omega^{NC}_A} \mu^{NC}_{\sigma}(s)\xi^{NC}_{m}(s)(k)
&= \sum_{\tilde{s} \in \mathcal{E}(X({\bf C}/\sim))}{\tilde{d}_{\tilde{\sigma}}}(\tilde{s}) \delta_{\tilde{s}|_{[m]}([m]),[k]}\\
&= \sum_{s \in \mathcal{E}(X({\bf C}))}d_{\sigma}(s) \delta_{s|_{m}(m),k}\\
&= d^A_{\sigma, m}(k)
\end{align*}

The first equality holds by unfolding definitions of the canonical representation. The second equality holds because $d_{\sigma}(s)$ is only nonzero on those sections $s$ that assign the same outcome to all equivalent measurements; therefore, we can extend the sum over $\mathcal{E}(X/\sim)$ to the sum over $\mathcal{E}(X)$. The last equality holds as both expressions are equal to $\sigma_{m}(s)(k)$.

This canonical ontological representation is by definition preparation non-contextual. On measurements, it is defined such that $m \sim m'$ implies $\xi^{NC}_{m}=\xi^{NC}_{m'}$; hence it is measurement non-contextual. 

It is left to determine under which conditions an operational theory can be realised by a non-contextual empirical theory. To this end, we generalise Theorem 8.1 of \cite{abramskybrandenburger}.

\begin{lemma}\label{lem:F'esssurj}
For every factorizable, non-contextual ontological representation $B$, there exists an empirical theory $A$ with a global section, such that $\R(A)$ and $B$ realise the same operational theory.
\end{lemma}

\begin{proof}
Let $B$ be a factorizable, measurement non-contextual ontological representation. The operational theory realised by $B$ induces an empirical theory $A$ where $X_A$ is given by the minimal elements of the preorder of joint measurements. By Lemma \ref{lem:nosig}, B is parameter independent. Every preparation $p \in P_B$ realises a state $\sigma_p$ with a global section $d_p$ for the sheaf of distributions induced by $A$. These are defined below for $r \in \cE(m)$ and $s \in \cE(X)$:

\begin{align}
\sigma_p(r) &:= \sum_{\lambda \in \Omega_B} \xi_{m}(\lambda)(r(m)) \mu_P(\lambda) \hspace{.1cm}
&d_{p}(s) &:=  \sum_{\lambda \in \Omega_B} \prod_{m \in X_A} \xi_{m}(\lambda)(s|_m(m)) \mu_P(\lambda)
\end{align}


We need to verify that $\R(A)$ and $B$ realise the same measurement statistics. This follows from the equalities below, where we denote the canonical ontological representation by $\Omega'_A, \mu'$, and $\xi'$.

\begin{align*}
\sum_{\tilde{s} \in \Omega'_A} \xi'_{m}(\tilde{s})(k) \mu'_p 
&=  \sum_{\tilde{s}\in \mathcal{E}(X_A/\sim)} \delta_{\tilde{s}|_{[m]}([m]),[k]} \bigg[\tilde{d}_{\tilde{\sigma}_P}(\tilde{s})\bigg]\\
&=  \sum_{s\in \mathcal{E}(X_A)} \delta_{s|_{m}(m),k} \bigg[\sum_{\lambda \in \Omega_B} \prod_{n \in X_A} \mu_p(\lambda) \xi_{n}(\lambda)(s|_{n}(n)) \bigg]\\
&= \sum_{\lambda \in \Omega_B} \xi_m(\lambda)(k)\bigg[\sum_{s \in \mathcal{E}(X_A\backslash m)}\hspace{.1cm} \prod_{n \in X_A\backslash m}  \xi_{n}(\lambda)(s|_n(n)) \bigg] \mu_p(\lambda) \\
&= \sum_{\lambda \in \Omega_B} \xi_m(\lambda)(k) \mu_p(\lambda)
\end{align*}

The first two equalities result from expanding definitions. For the third, we apply Fubini's theorem, split the sum and product, and rewrite the expression. The last equality holds because probability distributions sum to one over all the inputs.
\end{proof}

\begin{proof}[Proof of Theorem~\ref{thm:eq}]
For any empirical theory $A$, the canonical non-contextual ontological representation for $Op(A)$ is given by $R(A)$, which means that $1) \Rightarrow 2)$. The canonical ontological representation $R(A)$ is factorizable; therefore, $2) \Rightarrow 3)$. Finally, $3) \Rightarrow 1)$ holds by Lemma~\ref{lem:F'esssurj}.
\end{proof}

\begin{corollary}\label{cor:can}
For the class of perfectly predictable operational theories with a maximally mixed preparation, an operational theory is non-contextual iff its canonical ontological representation is non-contextual.
\end{corollary}

\begin{proof}
By Lemma~\ref{lem:od}, all preparation non-contextual ontological representations of operational theories in this class are outcome-deterministic. By Theroem 6 of~\cite{liangspekkenswiseman} that implies that all non-contextual ontological representations are factorizable. The result follows directly from Theorem~\ref{thm:eq}. 
\end{proof}

\section{A Categorical Isomorphism}\label{sec:catiso}
In this section, we show that the two formalisms can be used to represent the no-signalling world in equivalent ways. To give a formal proof, we use the mathematical framework of category theory.
We show that the correspondence between no-signalling empirical theories, operational theories and ontological representations discussed in the previous section gives rise to functors between suitable categories. In particular, there is an isomorphism between the categories $\Emp_{ns}$ of no-sigmalling empirical theories and $\Ot$ operational theories. This isomorphism maps non-contextual empirical theories to operational theories that admit a factorizable non-contextual ontological representation.


\subsection*{The Category of Empirical Theories}
\label{sec:catemp}
We will define  the category  $\Emp$  of empirical theories and  transformations that preserve contextuality and statistical equivalence. The category $\mathcal{E}mp$ is an extension of the category of empirical models introduced in~\cite{empcat}.

\begin{definition}
A transformation between empirical theories is given by a triple $f=(f^S,f^{\cM}, f^O)$ of maps between the set of states, the measurement cover and the set of outcomes, respectively. In addition, each assignment $C \mapsto f^{\cM}(C)$ consists of a functor $f^C: C \rightarrow f^{\cM}(C)$ of the subcategories of objects with an arrow to $C$ and $f^{\cM}(C)$, respectively.
\end{definition}

Note that if ${\bf C}_A$ and ${\bf C}_B$ are posets, $f$ is a simplicial map $\downarrow \mathcal{M}_A \rightarrow \downarrow \mathcal{M}_B$. We write $f$ for either component when it is clear from the context which one we mean. If a transformation satisfies the following equation, we can recover the statistical data of the domain from the statistical data of the image. 

\begin{equation}\label{eq:cp} \sigma_C(s) = \sum_{s' \circ f^{\bf C} = f^O \circ s} f^S(\sigma)_{f^{\bf C}(C)}(s') \hspace{2cm} \forall C \in \cM_A
\end{equation}

We will call such transformations {\bf contextuality preserving} due to the following Lemma.

\begin{lemma}\label{lem:globalsecsub}
Let $f:A \rightarrow B$ be a transformation of empirical theories that satisfies equation (\ref{eq:cp}) and let $\sigma$  be a state of $A$. If $\sigma$ does not admit a global section, then $f(\sigma)$ does not admit a global section.
\end{lemma}

\begin{proof}
Suppose that $f(\sigma) \in B$ has a global section $\nu \in D_R\mathcal{E}(f(X_A))$. This means that $\nu|_{C'} = f(\sigma)_{C'}$ for all $C' \in X$. This induces a global section for $\sigma$, given by $\mu(s) = \sum_{s' \circ f^{\mathcal{M}} = f^{O} \circ s} \nu|_{f(X_A)}(s')$ in $D_R(\mathcal{E}(X_A))$.

\end{proof}

When $A$ and $B$ are no-signalling theories, $f^{\cM}$ is simply a functor of categories ${\bf C}_A \rightarrow {\bf C}_B$.  Equation~\ref{eq:cp} then simplifies to the equation below.

\begin{equation}\label{eq:cp2} 
\sigma_m(s) =   f^S(\sigma)_{f^{C}(m)}(s') \hspace{1.5cm}  \forall f^O \circ s(m) = s' \circ f^{C}(m) \hspace{.1cm} \forall m \in X_A
\end{equation}

In addition to equation~\ref{eq:cp}, we require morphisms to preserve statistical equivalence:

\begin{equation}\label{eq:ep}
m \sim m' \hspace{.25cm} \Rightarrow \hspace{.25cm} f(m) \sim f(m') \hspace{1cm}
\sigma \sim \sigma' \hspace{.25cm} \Rightarrow \hspace{.25cm} f(\sigma) \sim f(\sigma')
\end{equation}



We can now prove the statement in Lemma~\ref{lem:gs}, that any global section gives rise to a global section defined on equivalence classes.

\begin{proof}[Proof of Lemma~\ref{lem:gs}] 
Consider the quotient map $A \xrightarrow{q} \tilde{A}$ and any inclusion map $\tilde{A} \xrightarrow{i} A$, which is defined as follows: $[C]$ is mapped to some representative $C$ such that $i^{\bf C}$ is a functor, and $\tilde{\sigma}$ is mapped to $\sigma$. It is easy to see that $q$ and $i$ are morphisms in $\mathcal{E}mp$. consequently, the proof follows from Lemma \ref{lem:globalsecsub}.
\end{proof}

\begin{remark}
Another way to define transformations between empirical theories is given in~\cite{ah}. Here, empirical models are defined in terms of Chu spaces and the function on states goes in the opposite direction. By that definition, contextuality of states would only be preserved by transformations that are surjective on states.
\end{remark}

In the rest of this paper we will restrict our attention to the subcategory $\Emp_{ns}$ of no-signalling empirical theories. We write $\Emp_{ns}^{NC}$ for the subcategory of non-contextual empirical theories.



\subsection*{The Category of Operational Theories}\label{sec:catOT}
Operational theories form a category $\mathcal{O}t$. Morphisms are tuples $f=(f^M,f^P, f^O): A \rightarrow B$, such that $f^M:M_A \rightarrow M_B$, $f^P:P_A \rightarrow P_B$ and $f^O: O_A \rightarrow O_B$ preserve outcome statistics and statistical equivalence:
\begin{equation} \label{def:otmorph}
d_{f^P(p),f^M(m)}(f^O(k)) = d_{p,m}(k)
\end{equation}

\begin{equation}\label{def:otmorph2}
m \sim m' \Rightarrow f(m) \sim f(m')  \hspace{1cm} p \sim p' \Rightarrow f(p) \sim f(p')
\end{equation}

This category is similar to the category of operational theories defined in~\cite{ah}.

\subsection*{The Category of Ontological Representations}\label{sec:catOR}
Objects in the category $\Or$ of ontological representations correspond to a pair of an ontological representation and its induced operational theory. Morphisms consist of triples of maps $(f, f^{\mu}, f^{\xi})$, where $f:A \rightarrow B$ is a morphism of operational theories, and $f^{\xi}: \xi \mapsto \xi'$ and $f^{\mu} : \mu \mapsto \mu'$ are functions of sets. We require that the image of $(f, f^{\mu},f^{\xi})$ realises the operational theory in the image of $f$. This means that the images of the elements of $\mu$ and $\xi$ coincide with the elements corresponding to the images of $f^P$ and $f^M$. We express this as  $f^{\mu}(\mu_{p}) = f^{\mu}(\mu)_{f^P(p)}$ and $f^{\xi}(\xi_m) = f^{\xi}(\xi)_{f^M(m)}$. In addition, one can deduce from equations \ref{def:or} and \ref{def:otmorph} that the equality below holds.

\begin{equation}\label{def:ORmorphism}
\sum_{\lambda \in \Omega_B} f^{\xi}(\xi_{m})(\lambda)(f^O(k)) f^{\mu}(\mu_p)(\lambda) =  \sum_{\lambda \in \Omega_A} \xi_m(\lambda)(k)\mu_p(\lambda)
\end{equation}

\begin{remark}
Note that the morphisms do not contain a component that maps between the sets of ontological values. This is because our goal is not to understand individual ontological representations, but to explore the existence of certain classes of ontological representations for operational theories. 
\end{remark}

There is a forgetful functor $G: \mathcal{O}R \rightarrow \Ot$ that maps each ontological representation to its corresponding operational theory. More precisely, it maps $(\Omega, \xi, \mu)$ to $(\{\mu_p\}_{p \in P}, \{\xi_m\}_{m \in M}, D, O)$, where $d
_{p,m}:=\sum_{\lambda \in \Omega}\mu_p(\lambda)\xi_{m}(\lambda)$. The elements $\mu_P$ and $\xi_M$ no longer represent distribution functions, but merely label the preparations and measurements.  

\begin{lemma}
Contextuality of operational theories is preserved by morphisms in $\mathcal{O}T$
\end{lemma}

\begin{proof}
Let $f:A \rightarrow B$ be a morphism of operational theories. Let $(\Omega_B,\{\mu_p\}_{p \in P_B}, \{\xi_{m}\}_{m \in M_B})$ be a non-contextual ontological representation of $B$. This induces an ontological representation \newline $(\Omega_B, \{\mu'_p\}_{p \in P_A}, \{\xi'_{m}\}_{m \in M_A})$, which is defined as $\mu'_p := \mu_{fp}$, $\xi'_{m}:=\xi_{fm}$. Non-contextuality of this ontological representation is guaranteed by the equivalence preservation condition on $f$. It follows by contradiction that when $A$ is contextual, $B$ must be contextual.
\end{proof}

\subsection*{The Isomorphism}\label{sec:iso}
The assignment $A \mapsto Op(A)$ of an operational theory to each no-signalling empirical theory described in Section~\ref{sec:canhidvar} gives rise to the functor below.
  
\begin{align*}
\begin{tikzpicture}[xscale=2.75, yscale=3]
\node[] (A) at (0,0){$\mathcal{E}mp_{ns}$};
\node[] (B) at (2,0){$\mathcal{O}T$};
\node (C) at (0,-.25){$A$};
\node (D) at (2,-.25){$Op(A)$};
\node (E) at (0,-.5){$(f^S, f^C, f^{O_A})$};
\node (F) at (2,-.5){$(f^M, f^P, f^{O_{Op(A)}})$};
\draw[->] (A) to node[above]{$Op$} (B);
\draw[|->] (C) to node[above]{} (D);
\draw[|->] (E) to node[above]{} (F);
\end{tikzpicture}
\end{align*}

Here, $f^M$ is defined as the assignment on objects of $f^{\cM}$, $f^P:=f^S$ and $f^{O_A} = f^{O_{Op(A)}}$, since $O_A = O_{Op(A)}$. To verify that this is well-defined on morphisms, one needs to check that equations~\ref{def:otmorph} and~\ref{def:otmorph2} hold. Since we only consider no-signalling empirical theories, this follows directly from equations~\ref{eq:cp2} and~\ref{eq:ep}. Functoriality is straightforward. 

\begin{prop}
The functor $\Emp_{ns} \xrightarrow{Op} \Ot$ is an isomorphism
\end{prop}

\begin{proof}
As discussed in section~\ref{sec:uni}, the assignment is bijective on objects. To see that it is injective on morphisms, note that the functor $f^{\cM}$ is completely determined by its assignments on objects, since ${\bf C}_A$ and ${\bf C}_B$ are thin categories.
Surjectivity follows from the fact that equations~\ref{eq:cp2} and~\ref{eq:ep} imply equations~\ref{def:otmorph} and~\ref{def:otmorph2}.
\end{proof}

We will show that the isomorphism $\Emp_{ns} \xrightarrow{Op} \Ot$ maps non-contextual empirical theories to operational theories that admit a factorizable non-contextual ontological representation. In order to do so, we first examine how the canonical ontological representation described in Section~\ref{sec:canhidvar} gives rise to a functor $\Emp^{NC}_{ns} \xrightarrow{R} \Or$. This functor maps each non-contextual empirical model to its canonical non-contextual ontological representation. It maps each morphism of empirical models to a morphism of ontological representations in an obvious way, such that the effect on the outcome statistics is the same in either model. It turns out that the composition of this functor with the forgetful functor $\Or \rightarrow \Ot$ equals $\Emp_{ns} \xrightarrow{Op} \Ot$ on the class of non-contextual empirical models.
\begin{prop}\label{prop:equiv}
For any choice of global sections, the assignment $A \mapsto R(A)$ defines an equivalence between the subcategory of non-contextual empirical theories and the subcategory of non-contextual, factorizable ontological representations. 
The image $Rf= (Rf, (Rf)^{\mu}, (Rf)^{\xi})$ of each morphism $f=(f^S, f^{\cM}, f^O)$ has the following components
\begin{align*}
Rf := Op(f) \hspace{1cm} 
& (Rf)^{\mu} ( \sigma_m )  :=   f\sigma_{f(m)}, \hspace{1cm}
& (Rf)^{\xi} (\xi_m)(s)(k)  := \delta_{s(f(m)),k}
\end{align*}
\end{prop}

\begin{proof}
We need to verify that for each $f:A \rightarrow B$ in $\Emp$, $Rf: R(A) \rightarrow R(B)$ is a well-defined morphism in the category of ontological representations. It is easy to see that since $f$ preserves statistical equivalence, $Rf$ does too.  By the following equations, $Rf$ also satisfies equation~\ref{def:ORmorphism}.

\begin{align}
\sum_{\lambda \in \Omega_{R(A)}}  \xi_{m}(\lambda)(k)\mu_{\sigma}(\lambda)
 &=  \sum_{\lambda \in \mathcal{E}(m)}  \delta_{\lambda({m}),k} \sigma_{m}(\lambda)\\ 
   &= \sum_{\lambda' \in \mathcal{E}(f(m))}  \delta_{\lambda'(f(m)),f^O(k)} f\sigma_{f(m)}(\lambda') \\
    &= \sum_{\lambda' \in \Omega_{R(B)}} Ff(\xi_{m})(\lambda')(f^O(k)) Rf(\mu_{\sigma})(\lambda') 
\end{align}

The equalities are obtained by unfolding definitions, application of equation~\ref{eq:cp2} and rewriting the summation.

 By Lemma \ref{lem:F'esssurj}, $R$ is essentially surjective on the subcategory of factorizable non-contextual ontological representations. We will show that the functor $R$ is injective on hom-sets. First of all, $\delta_{s(f(m)),k} = \delta_{s(g(m)),k)}$ for all $k \in O$ implies that $f^S=g^S$. Similarly, $R(f)^{\mu} = R(g)^{\mu}$ implies $f^C=g^C$. We will prove that $F$ is surjective on hom-sets. Let $(g, g^{\mu}, g^{\xi}): RA \rightarrow RB$ be a morphism in $\mathcal{O}R$. This corresponds to the morphism $g':A \rightarrow B$ in $\mathcal{E}mp$ with components $g'^{\cM}(m) = g^M(m)$ and $g'^{S}(\sigma)_{g'^{\cM}(m(s))} = g^P(\mu_{\sigma})(s)$ for $s \in \mathcal{E}(m)$. Since each $\xi$ in the image of $g$ is a delta function, it must be equal to $\delta_{\lambda(g(m)),k}$. Finally, to show that equation~\ref{def:otmorph2} holds, we take equation \ref{def:ORmorphism} and unfold the definitions of $\Omega_A$, $\Omega_B$, $g(\xi_M)$, and $\xi_M$. This gives us the equality below, which reduces to the second condition for transformations of empirical theories.

\begin{align}
\sum_{s \in \mathcal{E}(m)}  \delta_{s(m),k} \mu_{\sigma}(s) &=  \sum_{s \in \mathcal{E}(f(m))} \delta_{s(f(m)),k} g_{\mu}(\mu_{\sigma})(s)  
\end{align}
\end{proof}

\begin{theorem}\label{thm:iso}
The isomorphism $Op$ restricts to an isomorphism between the subcategory of non-contextual empirical theories and the subcategory of operational theories that do not admit a factorizable non-contextual ontological representation.
\end{theorem}

\begin{proof}
Note that the following diagram commutes, where we write $\Emp^{NC}$ and $\Or^{FNC}$ for the subcategories of non-contextual empirical theories and factorizable non-contextual ontological representations, respectively.

\begin{align*}
\begin{tikzpicture}[xscale=2.75, yscale=3]
\node[] (A) at (0,0){$\mathcal{E}mp$};
\node[] (B) at (2,0){$\mathcal{O}T$};
\node (C) at (0,-.5){$\mathcal{E}mp^{NC}$};
\node (D) at (2,-.5){$\mathcal{O}R$};
\node (E) at (1,-.75){$\Or^{FNC}$};
\draw[->] (A) to node[above]{$Op$} (B);
\draw[right hook->] (C) to  (A);
\draw[->] (D) to node[right]{$G$} (B);
\draw[->] (C) to node[above]{$R$} (E);
\draw[right hook->] (E) to (D);
\end{tikzpicture}
\end{align*}
\end{proof}

\begin{corollary}
For models with perfectly predictable measurements and a maximally mixed preparation, $R$ restricts to an isomorphism between non-contextual empirical theories and non-contextual operational theories.
\end{corollary}

\section{Non-factorizable representations and POVM's}\label{sec:postlude}
In general, equivalence-based measurement contextuality implies sheaf-theoretic contextuality, but not necessarily the other way around. The two formalisms coincide in any scenario where factorizability can be justified, such as in the following three cases:

\begin{itemize}
\item For theories with perfectly predictable measurements and a maximally mixed preparation, non-contextuality of an ontological representation implies factorizability.
\item Any non-local scenario rules out non-factorizable ontological representations, as these would violate local causality. 
\item In theories that do not contain joint measurements, the notion of factorizability is vacuous.
\end{itemize}

The following example, which is known as Specker's Triangle, shows that the two formalisms are not equal for all scenarios.

\begin{example}[Specker's Triangle]\label{ex:ex1}
There are three parties, $A,B,C$, that each conduct a measurement with two outcomes, $\{0,1\}$. It is possible for two parties to apply the measurement at the same time, but it is not possible to apply all three measurements simultaneously. The measurement statistics is such that for any joint measurement, the obtained outcome is $(0,1)$  half of the time, and $(1,0)$, half of the time. 
\end{example}

This scenario cannot be realised by sharp measurements in quantum mechanics. 
However, one can find a POVM for each joint measurement that margnalises to the required outcomes: $(0\cdot P_{0,0}, \frac{1}{2} \cdot P_{0,1}, \frac{1}{2} \cdot P_{1,0}, 0 \cdot P_{1,1})$, where $P_{i,j}$ is the projector onto outcome $(i,j)$. Note however, that this POVM can be classically realised, by flipping a coin to decide on outcome $(0,1)$ or $(1,0)$.

\begin{lemma}
The scenario in Example~\ref{ex:ex1} is contextual in the sheaf sense, but non-contextual in the equivalence-based sense
\end{lemma}

\begin{proof}
The marginal probabilities for each of the individual measurements are $\frac{1}{2}$ for either of the outcomes. It follows that all measurements are statistically equivalent, and hence, should not be distinguishable on the ontological level. This means that we can define the set $\Omega:= \{*\}$ to be a singleton set. We set $\mu_p(*)=1$ for any preparation of this scenario, $\xi_m(*)(0) = \xi_m(*)(1)= \frac{1}{2}$ for each of the elementary measurements, and $\xi_m(*)(0,1) = \xi_m(*)(1,0)= \frac{1}{2}$, for each of the joint measurements. 
On the other hand, it is not possible to define a factorizable non-contextual ontological representation. It is easy to see this, since without loss of generality, any global section of measurement outcomes to the presheaf describing this scenario must assign the same outcome to measurement $A$ and $B$. But that means that it does not marginalise to an admissible outcome for the joint measurement of $A$ and $B$. 
\end{proof}

For a complete comparison of the two notions, a better understanding of unsharp measurements is required. Another point of consideration is the extent to which the functors respect additional assumptions.  We have shown that for all known examples of contextuality conditional to assumptions in the equivalence-based framework, the two notions coincide. However, this may not be the case in general. Ideally, one would like to have a specification of the class of scenarios and assumptions for which the formalisms are different. We leave this for future work.

\section*{Acknowledgements}
I would like to thank Samson Abramsky for his supervision; Matty Hoban, Jonathan Barrett, Rui Soares Barbosa, Kohei Kishida and Ravi Kunjwal for helpful discussions. Support from the EPSRC Doctoral Training Partnership is also gratefully acknowledged.

\bibliographystyle{eptcs}
\bibliography{AEPC}{}

\end{document}